\pdfoutput=1
\documentclass[12pt,a4paper,reqno]{amsart}

\usepackage{amssymb}
\usepackage{amsmath}
\usepackage{amsthm}
\usepackage{thmtools}
\usepackage{mathtools}
\usepackage{mathabx}
\usepackage{mathrsfs}
\usepackage{soul}
\usepackage{bm}
\numberwithin{equation}{section}
\usepackage{framed}
\usepackage[bottom]{footmisc}
\usepackage[utf8]{inputenc}
\usepackage[colorlinks=false,hidelinks]{hyperref}
\usepackage[nameinlink,capitalize]{cleveref}
\usepackage{enumitem}
\usepackage{graphicx}
\usepackage[dvipsnames]{xcolor}
\usepackage{relsize}
\usepackage{exscale}
\usepackage{floatrow}
\usepackage{pdflscape}
\usepackage{afterpage}
\usepackage[justification=centering]{caption}
\usepackage{rotating}
\usepackage{pdflscape}
\usepackage{tabulary}
\usepackage{caption}
\usepackage{subcaption}

\theoremstyle{plain}

\newtheorem{theorem}[subsection]{Theorem}

\theoremstyle{definition}
\newtheorem{definition}[theorem]{Definition}
\newtheorem{example}[theorem]{Example}

\makeatletter
\newcommand{\vast}{\bBigg@{4}}
\newcommand{\Vast}{\bBigg@{5}}
\makeatother

\makeatletter
\newcommand*\bigcdot{\mathpalette\bigcdot@{.5}}
\newcommand*\bigcdot@[2]{\mathbin{\vcenter{\hbox{\scalebox{#2}{$\m@th#1\bullet$}}}}}
\makeatother

\usepackage[square,numbers]{natbib}
\usepackage{amsaddr}
\usepackage{etoolbox}

\makeatletter
\patchcmd{\@maketitle}
{\ifx\@empty\@dedicatory}
{\ifx\@empty\@date \else {\vskip3ex \centering\@date\par\vskip1ex}\fi
	\ifx\@empty\@dedicatory}
{}{}
\patchcmd{\@maketitle}
{\ifx\@empty\@date\else \@footnotetext{\@setdate}\fi}
{}{}{}
\makeatother

\renewcommand{\leq}{\leqslant}
\renewcommand{\geq}{\geqslant}

\def\vs{\vspace{0.2cm}}
\def\ni{\noindent}
\def\emph#1{{\it #1}}
\def\textbf#1{{\bf #1}}

\usepackage{subcaption}
\begin{document}
	
\title{Conditional precedence orders for stochastic comparison of random variables}

\author{Sugata Ghosh and Asok K. Nanda}
\address{Department of Mathematics and Statistics\\
	Indian Institute of Science Education and Research Kolkata\\
	Mohanpur 741246, India}
\email{sg18rs017@iiserkol.ac.in}
\email{asok@iiserkol.ac.in}

\subjclass[2010]{60E15}
\keywords{Stochastic order, Connex property, Statistical dependence}

\allowdisplaybreaks
\raggedbottom

\begin{abstract}
	Most of the stochastic orders for comparing random variables, considered in the literature, are afflicted with two main drawbacks: $(i)$ lack of {\it connex} property and $(ii)$ lack of consideration of any dependence structure between the random variables. Both these drawbacks can be overcome at the cost of transitivity with the stochastic precedence order, which may seem to be a good choice in particular when only two random variables are under consideration, a situation where the question of transitivity does not arise. In this paper, we show that even under such favorable conditions, stochastic precedence order may direct to misleading conclusion in certain situations and develop variations of the order to address the phenomenon.
\end{abstract}

\maketitle
\markleft{\uppercase{Conditional precedence orders}}
\markright{\uppercase{Sugata Ghosh and Asok K. Nanda}}

\section{Introduction}\label{c-section-introduction}
	Randomness is an unavoidable phenomenon that occurs in any scientific study. It can arise either naturally, or out of our limitations in taking precise measurements. In most scientific fields, we face the problem of comparing various quantitative features of different schemes, species or categories. These measurements are inevitably contaminated by random errors that we cannot control, but can model them using appropriate probabilistic framework. Stochastic orders are useful tools to compare random variables in a systematic way. The concept finds applications in several areas of studies including statistics, probability, actuarial science, operations research, risk management and wireless communications. Let $X$ and $Y$ be two continuous random variables with respective cumulative distribution functions (cdf) $F_X$ and $F_Y$, probability density functions (pdf) $f_X$ and $f_Y$, and hazard rate functions $r_X$ and $r_Y$. Below we give definitions of four widely applied stochastic orders.
	
	\begin{definition}\label{c-definition-stochastic-order}
		$X$ is said to be less than $Y$ in
		\begin{enumerate}[label=\textnormal{(\roman*)}]
			\item {\it usual stochastic order} (denoted by $X \leq_{\textnormal{st}} Y$) if $F_X(x) \geq F_Y(x)$, for every $x \in \mathbb{R}$.
			\item {\it hazard rate order} (denoted by $X \leq_{\textnormal{hr}} Y$) if $r_X(x) \geq r_Y(x)$, for every $x \in \mathbb{R}$.
			\item {\it likelihood ratio order} (denoted by $X \leq_{\textnormal{lr}} Y$) if $f_Y(x)/f_X(x)$ is increasing in $x \in \mathbb{R}$.
			\item {\it mean residual life order} (denoted by $X \leq_{\textnormal{mrl}} Y$) if, for every $x \in \mathbb{R}$,
			\[
			\frac{\int_x^{\infty} \bar{F}_X(t) dt}{\bar{F}_X(x)} \geq \frac{\int_x^{\infty} \bar{F}_Y(t) dt}{\bar{F}_Y(x)},
			\]
			where $\bar{F}_X(t)=1-F_X(t)$ and $\bar{F}_Y(t)=1-F_Y(t)$, for every $t \in \mathbb{R}$. \hfill $\blacksquare$
		\end{enumerate}
	\end{definition}
	
	For more details on various stochastic orders, we refer to \citet{MS_2002} and \citet{SS_2007}. One can immediately check that the stochastic orders, defined above, are all partial orders. Also, it can be verified that likelihood ratio order implies hazard rate order, which in turn implies both usual stochastic order and mean residual life order. However, these stochastic orders lack two crucial aspects which one may expect from a valid ordering procedure for random variables.\vs
	
	\begin{enumerate}
		\item The {\it connex} property, i.e. given two random variables $X$ and $Y$ it may happen that none of $X \leq_{\textnormal{A}} Y$ and $Y \leq_{\textnormal{A}} X$ hold, where $A$ is any partial stochastic order.\vs
		
		\item The orders do not consider any possible dependence structure between the random variables involved. The reason behind this limitation is that these orders are defined based on respective marginal distributions of the random variables. Hence, any dependence structure, encoded in the joint distribution gets lost.\vs
	\end{enumerate}

	The first problem can be solved by upgrading to total orders, which unfortunately comes with the cost of throwing away most of the information available in the marginal distributions. An example of such order in the {\it mean order}, defined as follows.
	\begin{definition}\label{c-definition-mean-order}
		Let $X$ and $Y$ be two random variables with finite mean, i.e. $E\vert X \vert<\infty$ and $E\vert Y \vert<\infty$. Then, $X$ is said to be less than $Y$ in {\it mean order} (denoted by $X \leq_{\textnormal{mean}} Y$) if $E(X) \leq E(Y)$. \hfill $\blacksquare$
	\end{definition}

	This order still does not capture any possible dependence between $X$ and $Y$. \citet{AKS_2002} introduced the following stochastic order, that satisfies the connex property and captures dependence between the random variables in a certain way, at the cost of giving up the transitivity property.
	\begin{definition}\label{c-definition-stochastic-precedence-order}
		$X$ is said to be less than $Y$ in {\it stochastic precedence order} (denoted by $X \leq_{\textnormal{sp}} Y$) if $P(X \leq Y) \geq 1/2$. \hfill $\blacksquare$
	\end{definition}
	Note that $X$ is said to be equal to $Y$ in stochastic order $A$ (denoted by $X =_{\textnormal{A}} Y$) if $X \leq_{\textnormal{A}} Y$ and $Y \leq_{\textnormal{A}} X$, where $A$ can be any valid stochastic order, such as the ones defined above. The analysis of the quantity $P(X \leq Y)$, which quantifies the dominance of $Y$ over $X$, was first considered in \citet{B_1956}. It has extensive applications in stress-strength analysis (see \citet{K_2003} for more details). When $X$ and $Y$ are independent, usual stochastic order implies stochastic precedence order. The lack of transitivity becomes crucial when more than two random variables are involved in the problem of stochastic comparison. However, this issue does not arise in the problem of comparing only two random variables. So, it may seem that, in a problem of stochastic comparison of two random variables, stochastic precedence order is the best option available, in the sense that it enjoys the desirable connex property and uses the joint distribution of $X$ and $Y$ (instead of only the respective marginal distributions), thereby takes into account any underlying dependence structure. However, the next example, although somewhat contrived, refutes this intuition.
	\begin{example}\label{c-example-conditional-1}
		Consider a simple gambling scenario where a gambler has to choose between two schemes. If he chooses scheme $A$, then a biased coin with probability of head $0.6$ is tossed. If head comes up, then the gambler wins $1000$ points, whereas if tail comes up, then he returns empty-handed. If he chooses scheme $B$, then he wins $999$ points irrespective of the outcome of the toss of the above mentioned coined.\vs
		\begin{center}
			\begin{tabular}{ |c|c|c| } 
				\hline
				$\downarrow$ Outcome / Scheme $\rightarrow$ & A & B \\
				\hline 
				Head (probability 0.6) & 1000 & 999 \\
				\hline
				Tail (probability 0.4) & 0 & 999 \\ 
				\hline
			\end{tabular}
		\end{center}\vs
		
		Let $X$ and $Y$ respectively denote the random variable that denote points earned by the gambler if he opts for scheme $A$ and scheme $B$. We intend to stochastically compare $X$ and $Y$. Before we employ any stochastic order, we note that it is always sensible to go for scheme $B$ with guaranteed profit, instead of opting scheme $A$ where one faces the risk of winning nothing for only a single more point to earn. A valid stochastic order should reflect this. Now, it is easy to check that the respective cdfs of $X$ and $Y$ cross each other and hence the partial orders are of no help in this situation. Turning to stochastic precedence order, we see that $P(Y \leq X)=0.6$ and hence $Y \leq_{\textnormal{sp}} X$, prompting one to choose scheme $A$ that contradicts common sense. Turning to mean order, we see that $E(X)=600$ and $E(Y)=999$. Thus $X \leq_{\textnormal{mean}} Y$, which agrees on the sensible conclusion. The example clearly shows that a gambler, armed with stochastic precedence order only, may run into trouble in the setup described above. \hfill $\blacksquare$
	\end{example}\vs

	There are two principal ways of comparing two random variables $X$ and $Y$. One way is to compare them through the original probability space $(\Omega,\mathscr{F},P)$ that they are defined on. The other way is through their distributions based on the derived probability spaces $(\mathbb{R},\mathscr{B},P_X)$ and $(\mathbb{R},\mathscr{B},P_Y)$, induced by the random variables. Observe that there are two main aspects of the comparative behavior of two random variables when we compare them based on the original probability space. One is the region of the sample space where $X \leq Y$, i.e. $\{\omega \in \Omega : X(\omega) \leq Y(\omega)\}$. The probability measure of this region can be thought of as a measure of dominance of $Y$ over $X$ and forms the basis of stochastic precedence order. The other aspect is the difference between $X(\omega)$ and $Y(\omega)$ at a sample point $\omega \in \Omega$. The mean order is entirely dependent on the sign of the difference $E(X)-E(Y)$, which can be written as
	\[
	\int_{\Omega} \left\{X(\omega)-Y(\omega)\right\}\,dP(\omega).
	\]
	Clearly, this quantity is a probabilistic accumulation of the difference between $X$ and $Y$ over the entire sample space, not just any specific sub-region inside it. Thus, the stochastic precedence order and the mean order stand as the purest representations of the two aspects described above in the sense that they consider only one aspect and completely ignores the other. In the next two sections, we introduce two stochastic orders that use both these aspects to come to a conclusion on preferring one random variable over the other. We shall show that there are situations when these new orders can provide better choices than traditional stochastic orders that employ a far limited information about the random variables.\vs

\section{Conditional $\mathcal{L}_1$ precedence order}\label{c-section-conditional-1}
	The problem with stochastic precedence order is that it involves the joint distribution only through the quantity $P(X \leq Y)$, which is only a tiny bit of information derived from the joint distribution. In fact, for a given sample point $\omega \in \Omega$, \cref{c-definition-stochastic-precedence-order} only considers whether $X(\omega) \leq Y(\omega)$ or not, and remains completely ignorant of the quantity $\vert X(\omega)-Y(\omega) \vert$. In this section, we develop a stochastic order that takes this difference into consideration.\vs
	
	The general approach is to consider an appropriate statistical distance between two random variables $X$ and $Y$ and partition the measure into two parts, one corresponding to the event $X<Y$ and the other corresponding to the event $X>Y$ (the event $X=Y$ does not contribute to any valid distance between $X$ and $Y$). The stochastic orders considered in this work are constructed on the basis of comparison of these two terms, in terms of their contribution to the sum. If the first term is bigger (resp. smaller) than the second, we take it as an indication that $X$ is smaller (resp. larger) than $Y$. If the two terms are equal, then we consider the two random variables equal in order. First we consider the following distance.
	\[
	\mathcal{L}_1\left(X,Y\right)=E\left(\left\vert X-Y \right\vert\right).
	\]
	Let us assume that $\mathcal{L}_1\left(X,Y\right)<\infty$, which is guaranteed if both the random variables have finite first order moments. Observe that
	\[
	\mathcal{L}_1\left(X,Y\right)=E\left(Y-X | X<Y\right)P\left(X<Y\right)+E\left(X-Y | X>Y\right)P\left(X>Y\right).
	\]
	This partition forms the basis of the following stochastic order.
	\begin{definition}\label{c-definition-conditional-mean-precedence-order}
		$X$ is said to be less than $Y$ in {\it conditional $\mathcal{L}_1$ precedence order} (denoted by $X \leq_{\textnormal{cp-}\mathcal{L}_1} Y$) if
		\begin{equation}\label{c-eq-conditional-mean-precedence-order}
		E\left(Y-X | X<Y\right)P\left(X<Y\right) \geq E\left(X-Y | X>Y\right)P\left(X>Y\right).
		\end{equation}
		$X$ and $Y$ are said to be equal in conditional $\mathcal{L}_1$ precedence order (denoted by $X =_{\textnormal{cp-}\mathcal{L}_1} Y$) if $X \leq_{\textnormal{cp-}\mathcal{L}_1} Y$ and $Y \leq_{\textnormal{cp-}\mathcal{L}_1} X$. Note that if $\mathcal{L}_1\left(X,Y\right)>0$, then \eqref{c-definition-conditional-mean-precedence-order} can be equivalently written as
		\begin{equation*}\label{c-eq-conditional-mean-precedence-order-1-2}
			\left\{E\left(\left\vert X-Y \right\vert\right)\right\}^{-1} E\left(Y-X | X<Y\right)P\left(X<Y\right) \geq 1/2. \tag*{$\blacksquare$}
		\end{equation*}
	\end{definition}
	
	In \cref{c-example-conditional-1}, we have $E\left(Y-X | X<Y\right)P\left(X<Y\right)=999 \times 0.4=399.6$ and $E\left(X-Y | X>Y\right)P\left(X>Y\right)=1 \times 0.6=0.6$. Thus $X \leq_{\textnormal{cp-}\mathcal{L}_1} Y$, which agrees on the sensible conclusion and disagrees on the conclusion drawn from stochastic precedence order. What it succeeds in incorporating (and stochastic precedence order fails to capture) is that even though the event $X>Y$ occurs with a high probability, the difference between $X$ and $Y$ is extremely small, making its contribution to $\mathcal{L}_1\left(X,Y\right)$ almost negligible. On the other hand, despite the fact that the event $X<Y$ occurs with low probability, the difference between $X$ and $Y$ is so huge in this event that it constitutes almost the entirety ($\approx 99.8\%$) of $\mathcal{L}_1\left(X,Y\right)$. Interestingly, the next result shows that the conditional $\mathcal{L}_1$ precedence order coincides with the mean order if the random variables under consideration are independent.
	\begin{theorem}\label{c-theorem-L-1-mean-equivalence-independence}
		Let $X$ and $Y$ be independent random variables. Then $X \leq_{\textnormal{cp-}\mathcal{L}_1} Y$ if and only if $E(X) \leq E(Y)$.
	\end{theorem}

	\begin{proof}
		We prove the case when $X$ and $Y$ are continuous. If the random variables are discrete, then the proof is similar. Let us denote the respective cdfs of $X$ and $Y$ by $F_X$ and $F_Y$. By hypothesis, $X$ and $Y$ are independent. Then
		\begin{align*}
			&\,\,\,\,\,\,\,E\left(Y-X \vert X<Y\right)P\left(X<Y\right)\\
			&=\int\limits_{x<y} \left(y-x\right) dF_X(x) dF_Y(y)\\
			&=\int\limits_{-\infty}^{\infty} \int\limits_{-\infty}^{y} y dF_X(x) dF_Y(y)-\int\limits_{-\infty}^{\infty} \int\limits_{x}^{\infty} x dF_X(x) dF_Y(y)\\
			&=\int\limits_{-\infty}^{\infty} yF_X(y) dF_Y(y)-\int\limits_{-\infty}^{\infty} x \left(1-F_Y(x)\right) dF_X(x)\\
			&=\int\limits_{-\infty}^{\infty} x\left\{F_X(x) dF_Y(x)+F_Y(x) dF_X(x)\right\}-\int\limits_{-\infty}^{\infty} x dF_X(x)\\
			&=\int\limits_{-\infty}^{\infty} x d\left(F_X(x)F_Y(x)\right)-E(X).
		\end{align*}
	
		Similarly, we compute that
		\[
		E\left(X-Y \vert X>Y\right)P\left(X>Y\right)=\int\limits_{-\infty}^{\infty} x\,d\left(F_X(x)F_Y(x)\right)-E(Y).
		\]
		
		The proof now follows from the observation that $E\left(Y-X \vert X<Y\right)P\left(X<Y\right) \geq E\left(X-Y \vert X>Y\right)P\left(X>Y\right)$ if and only if $E(X) \leq E(Y)$.
	\end{proof}\vs

	Despite the equivalence under independence, the conditional $\mathcal{L}_1$ precedence order clearly improves upon the mean order when the random variables are not independent, by considering their dependence structure. The next result shows how the order behaves when the random variables go through identical location-scale transformations. The proof is straightforward and hence omitted.
	
	\begin{theorem}\label{c-theorem-L-1-location-scale}
		Let $X$ and $Y$ be two random variables and let $a \in \mathbb{R}$. Also assume that $X \leq_{\textnormal{cp-}\mathcal{L}_1} Y$. Then $a+bX \leq_{\textnormal{cp-}\mathcal{L}_1} a+bY$ if $b>0$ and $a+bY \leq_{\textnormal{cp-}\mathcal{L}_1} a+bX$ if $b<0$. \hfill $\blacksquare$
	\end{theorem}

	In particular, the conditional $\mathcal{L}_1$ precedence order is reversed under taking negation of the random variables. Note that this result cannot be strengthened to the more general case of monotone transformations, as demonstrated in the following example.
	\begin{example}
		Let us consider any nondecreasing function $\phi:\mathbb{R} \to \mathbb{R}$ satisfying $\phi(0)=0$, $\phi(999)=1$ and $\phi(1000)=1000$. In the setup of \cref{c-example-conditional-1}, let the new profits of the gamblers, betting in scheme $A$ and scheme $B$ respectively, be $\phi(X)$ and $\phi(Y)$. The gambling procedure is tabulated as below. 
		\begin{center}
			\begin{tabular}{ |c|c|c| } 
				\hline
				$\downarrow$ Outcome / Scheme $\rightarrow$ & A & B \\
				\hline 
				Head (probability 0.6) & 1000 & 1 \\
				\hline
				Tail (probability 0.4) & 0 & 1 \\ 
				\hline
			\end{tabular}
		\end{center}\vs
		Now it is easy to verify that $\phi(Y) \leq_{\textnormal{cp-}\mathcal{L}_1} \phi(X)$, even though $X \leq_{\textnormal{cp-}\mathcal{L}_1} Y$.
	\end{example}\vs
	
\section{Conditional $\mathbf{K^*}$ precedence order}\label{c-section-conditional-2}
	Observe that if $E\left(\left\vert X-Y \right\vert\right)$ is infinite, then it is possible that both sides of the inequality \eqref{c-eq-conditional-mean-precedence-order} are infinite. In such a scenario, it is difficult to choose one random variable over the other. The only choices for a practitioner, using conditional $\mathcal{L}_1$ precedence order, are to either declare $X$ and $Y$ to be equal in order or deem the conditional $\mathcal{L}_1$ precedence order inconclusive in this situation.\vs
	
	Also observe that conditional $\mathcal{L}_1$ precedence order imposes great importance on the absolute difference between the random variables, which makes it a better choice than stochastic precedence order when $X-Y$ may take extreme values at one side of $0$ (even with small probability), while taking relatively small values at the other side. On the other hand stochastic precedence order completely ignores this aspect except only the sign of $X-Y$. In practice, however, one may want to strike a balance between the two. The need of such a balancing order is demonstrated in the next example.
	\begin{example}\label{c-example-conditional-2}
		Consider the same gambling scenario as in \cref{c-example-conditional-1}, with an even more biased coin having probability of head $0.9$. Also the gambler who opts for scheme $A$, in case head turns up, is awarded much more points than that in \cref{c-example-conditional-1}. The new deal is summarized as follows.\vs
		\begin{center}
			\begin{tabular}{ |c|c|c| } 
				\hline
				$\downarrow$ Outcome / Scheme $\rightarrow$ & A & B \\
				\hline 
				Head (probability 0.9) & 1100 & 999 \\
				\hline
				Tail (probability 0.1) & 0 & 999 \\ 
				\hline
			\end{tabular}
		\end{center}\vs
		
		Let $X$ and $Y$ respectively denote the points earned by the gambler if he opts for scheme $A$ and scheme $B$. This time the conclusion out of common sense is not as obvious as in the previous example. One can still walk away with $999$ points by choosing the safe scheme $B$. But as both the probability of head and the profit in scheme $A$ if head occurs go significantly higher, one is tempted to go for the risky scheme $A$. If we take the way of stochastic precedence order, i.e. make a decision based on the probability $P(X \leq Y)$, then the favor for scheme $A$ is overwhelmingly high. It can be easily verified that mean order retains its stance on favoring scheme $B$, as does conditional $\mathcal{L}_1$ precedence order. However, one may choose to refrain from acting based on the probabilities only, at the same time choosing not to give the realized values of $X-Y$ too much importance. In such a scenario, difference between the random variables maybe scaled down before being considered in the decision making process.
	\end{example}\vs
	
	To construct a stochastic order that ensures a conclusion for any pair of random variables, as well as scales down the effect of the difference between their realized values, we consider the {\it Ky Fan} metric, given by
	\[
	\mathbf{K^*}\left(X,Y\right)=E\left(\frac{\left\vert X-Y \right\vert}{1+\left\vert X-Y \right\vert}\right).
	\]
	It is well-known that this metric metrizes convergence in probability on the space of real-valued random variables (see \citet[Theorem $3.5$]{D_1976}). Also, it is clear from the definition that this metric can take values only in $[0,1)$. Consider the following partition, which motivates the next definition.
	\[
	\mathbf{K^*}\left(X,Y\right)=E\left(\frac{Y-X}{1+Y-X} \bigg| X<Y\right)P\left(X<Y\right)+E\left(\frac{X-Y}{1+X-Y} \bigg| X>Y\right)P\left(X>Y\right).
	\]
	
	\begin{definition}\label{c-definition-conditional-k-precedence-order}
		$X$ is said to be less than $Y$ in {\it conditional $\mathbf{K^*}$ precedence order} (denoted by $X \leq_{\textnormal{cp-}\mathbf{K^*}} Y$) if
		\begin{equation}\label{c-eq-conditional-k-precedence-order}
			E\left(\frac{Y-X}{1+Y-X} \bigg| X<Y\right)P\left(X<Y\right) \geq E\left(\frac{X-Y}{1+X-Y} \bigg| X>Y\right)P\left(X>Y\right).
		\end{equation}
		$X$ and $Y$ are said to be equal in conditional $\mathbf{K^*}$ precedence order (denoted by $X =_{\textnormal{cp-}\mathbf{K^*}} Y$) if $X \leq_{\textnormal{cp-}\mathbf{K^*}} Y$ and $Y \leq_{\textnormal{cp-}\mathbf{K^*}} X$. Note that if $\mathbf{K^*}\left(X,Y\right)>0$, then \eqref{c-definition-conditional-k-precedence-order} can be equivalently written as
		\begin{equation*}\label{c-eq-conditional-k-precedence-order-1-2}
			\left\{E\left(\frac{\left\vert X-Y \right\vert}{1+\left\vert X-Y \right\vert}\right)\right\}^{-1} E\left(\frac{Y-X}{1+Y-X} \bigg| X<Y\right)P\left(X<Y\right) \geq 1/2.
		\end{equation*}
	\end{definition}
	
	\vspace*{0.2cm}
	\ni In \cref{c-example-conditional-1}, we have
	\[
	E\left(\frac{Y-X}{1+Y-X} \bigg| X<Y\right)P\left(X<Y\right)=0.3996
	\]
	and
	\[
	E\left(\frac{X-Y}{1+X-Y} \bigg| X>Y\right)P\left(X>Y\right)=0.3.
	\]
	
	\vspace*{0.2cm}
	\ni Hence $X \leq_{\textnormal{cp-}\mathbf{K^*}} Y$, which contradicts stochastic precedence order and agrees on the sensible conclusion. However, in \cref{c-example-conditional-2}, we compute
	\[
	E\left(\frac{Y-X}{1+Y-X} \bigg| X<Y\right)P\left(X<Y\right)=0.0999
	\]
	and
	\[
	E\left(\frac{X-Y}{1+X-Y} \bigg| X>Y\right)P\left(X>Y\right)=0.8912.
	\]

	\vspace*{0.2cm}
	\ni Thus in this case, conditional $\mathbf{K^*}$ precedence order contradicts both mean order and conditional $\mathcal{L}_1$ precedence order, giving an overwhelming preference on $X$ over $Y$, leading the gambler to a risky strategy with $90\%$ chance of great profit and $10\%$ chance of catastrophic loss.\vs
	
	\ni{\it Remark.} The reason behind the contradiction between conditional $\mathcal{L}_1$ precedence order and conditional $\mathbf{K^*}$ precedence order in certain situations is that $\mathbf{K^*}$ distance uses the scaled version $\left\vert x-y \right\vert/\left(1+\left\vert x-y \right\vert\right)$ instead of the usual distance $\left\vert x-y \right\vert$. Since $\mathbf{K^*}$ distance is bounded in $[0,1)$, it undermines the effect of huge differences, which may turn crucial in certain situations. For instance, in \cref{c-example-conditional-1}, the difference between earned points in the two schemes is huge ($999$ points) if tail occurs. However, the scaling of $\mathbf{K^*}$ distance reduces it down to $0.999$, whereas the difference of $1$ point in case head turns up, is reduced only to $0.5$.\vs
	
\section{Conclusion}\label{c-section-conclusion}
	With several stochastic orders available in the literature, the practitioners face the problem of choosing one among the lot, which can be tricky as they often contradict among themselves. The next example, a variant of which is available in \citet{B_1972}, demonstrates such a contradiction.
	\begin{example}\label{c-example-contradiction}
		Let $0<\epsilon<1$. Consider two random variables $X$ and $Y$ with joint probability density function given by
		\[
		f_{X,Y}(x,y)=
		\begin{cases*}
			\frac{1-\epsilon}{\epsilon\left(1-\frac{\epsilon}{2}\right)} & if $(x,y) \in (0,1) \times (0,1)$ and $0 \leq x-y \leq \epsilon$,\\
			\frac{2}{\epsilon} & if $(x,y) \in (0,1) \times (0,1)$ and $y-x>1-\epsilon$,\\
			0 & otherwise.
		\end{cases*}
		\]
		It is easy to verify that $P(X \leq Y)=\epsilon^2/2$ and $F_X(x) \geq F_Y(x)$, for every $x \in \mathbb{R}$. Hence, we have $Y \leq_{\textnormal{sp}} X$, but $X \leq_{\textnormal{st}} Y$. \hfill $\blacksquare$
	\end{example}
	
	The choice of the order depends on the particular situation at hand and this is where the judgment of the practitioner comes into play. In this paper, we have introduced two new stochastic orders that work better than the traditional stochastic orders in certain situations. The current section gives a rough prescription on how to choose among stochastic precedence order, mean order, conditional $\mathcal{L}_1$ precedence order and conditional $\mathbf{K^*}$ precedence order.\vs
	
	It is generally agreed that in the context of comparing two random variables, the conclusion drawn from a stochastic order must reflect common sense whenever it is obvious. A counter-intuitive stochastic order that contradicts common sense does not serve the purpose of valid comparison of random variables. Unlike precise mathematical criteria, common sense is highly subjective. It depends on the priorities of the practitioner, dictated by the situation that he is dealing with.\vs
	
	For instance, suppose the situation is only concerned with whether or not the realized value of one random variable is bigger than that of the other and the difference between the two realized values are of no consequence. This is usually the case in a win-loss scenario where a win brings same fortune irrespective of its margin. For example, in \cref{c-example-conditional-1}, if we impose the rule that whoever gathers most points will win the game and the other will lose, with the difference of points playing no role whatsoever, then one must go for stochastic precedence order. On the other extreme, if two players, one always betting on scheme $A$ and the other on scheme $B$, keep on playing the game for a large number of times and whoever gathers the most points in sum wins the game, then one must employ the mean order.\vs
	
	Now, let us consider the situation when the absolute difference of the random variables does play a role. Then one cannot rely on stochastic precedence order. Again suppose that the gamble is not repeated for a large number of times, i.e. it does not wait for the unlikely event of high impact to occur. Then using mean order also has no point. Now we have conditional $\mathcal{L}_1$ precedence order and conditional $\mathbf{K^*}$ precedence order as possible candidates to opt from. It is not possible to mathematically quantify the importance of the difference $\left\vert X(\omega)-Y(\omega) \right\vert$ in the process of preferring one random variable over the other and hence a practitioner has to use his judgment to choose the most appropriate stochastic order for the particular situation at hand. Roughly speaking, if the absolute difference between the random variables is very important (as in \cref{c-example-conditional-1}), then one should employ $\mathcal{L}_1$ precedence order. However, if it is somewhat important (as in \cref{c-example-conditional-2}, especially for an ambitious gambler with risk-taking mentality), then one should opt for conditional $\mathbf{K^*}$ precedence order.\vs
	
	The following table summarizes the preference between $X$ and $Y$ (whichever is bigger than the other) according to stochastic precedence order, mean order, conditional $\mathcal{L}_1$ precedence order and conditional $\mathbf{K^*}$ precedence order in the respective setups of \cref{c-example-conditional-1} and \cref{c-example-conditional-2}.\vs
	
	\begin{center}
		\begin{tabular}{ |c|c|c| } 
			\hline
			$\downarrow$ Stochastic order / Example $\rightarrow$ & \cref{c-example-conditional-1} & \cref{c-example-conditional-2} \\
			\hline 
			Stochastic precedence order & $X$ & $X$ \\
			\hline 
			Mean order & $Y$ & $Y$ \\
			\hline
			Conditional $\mathcal{L}_1$ precedence order & $Y$ & $Y$ \\
			\hline
			Conditional $\mathbf{K}^*$ precedence order & $Y$ & $X$ \\ 
			\hline
		\end{tabular}
	\end{center}\vs

	It should be noted that conditional $\mathbf{K}^*$ precedence order is just one specific way to control the effect of the difference between random variables to the process of ordering them. The practitioner may fine-tune the balance by choosing appropriate scale-down of the difference between random variables, according to the requirement of the situation at hand.
	
	\vs
	{\bf Acknowledgement.} The first author is supported by Senior Research Fellowship from University Grants Commission, Government of India.\newline
	
	\baselineskip=16pt
	
	\bibliographystyle{newapa}
	\bibliography{references}
\end{document}